\documentclass[journal]{IEEEtran}
\usepackage{amsthm,amssymb,amsmath,rangecite}
\usepackage{commath,amsmath,graphicx,epstopdf,amsthm,amssymb,float,cite,color,array}
\usepackage[ruled,vlined]{algorithm2e}
\newtheorem{Theorem}{Theorem}

\newtheorem{proposition}[Theorem]{Proposition}

\newtheorem{remark}{Remark}

\usepackage{lscape}
\usepackage{cite}
\usepackage{url}
\usepackage{hyperref}

\newcommand{\onevec}{{\bf{1}}}

\newcommand{\define}{\stackrel{\triangle}{=}}

\newcommand{\be}{\begin{equation}}
\newcommand{\ee}{\end{equation}}
\newcommand{\beqna}{\begin{eqnarray}}
\newcommand{\eeqna}{\end{eqnarray}}

\begin{document}
\title{Bayesian Methods for Multiple Change-Point Detection with Reduced Communication}
\author{Eyal~Nitzan,~\IEEEmembership{Member,~IEEE,}
Topi~Halme,~\IEEEmembership{Student~Member,~IEEE,}
and~Visa~Koivunen,~\IEEEmembership{Fellow,~IEEE}
\thanks{This work was partly supported by the Academy of Finland projects: (1) Statistical Signal Processing Theory and Computational Methods for Large Scale Data Analysis (2) WiFiUS project: Secure Inference in the Internet of Things.}
\thanks{{\footnotesize{E. Nitzan, T. Halme, and Visa Koivunen are with the Department of Signal Processing and Acoustics, Aalto University, Espoo, Finland, e-mail: eyal.nitzan@aalto.fi, topi.halme@aalto.fi, and visa.koivunen@aalto.fi.}}}
}

\maketitle
\nopagebreak

\begin{abstract}
In many modern applications, large-scale sensor networks are used to perform statistical inference tasks. In this paper, we propose Bayesian methods for multiple change-point detection using a sensor network in which a fusion center (FC) can receive a data stream from each sensor. Due to communication limitations, the FC monitors only a subset of the sensors at each time slot. Since the number of change points can be high, we adopt the false discovery rate (FDR) criterion for controlling the rate of false alarms, while minimizing the average detection delay (ADD). We propose two Bayesian detection procedures that handle the communication limitations by monitoring the subset of the sensors with the highest posterior probabilities of change points having occurred. This monitoring policy aims to minimize the delay between the occurrence of each change point and its declaration using the corresponding posterior probabilities. One of the proposed procedures is more conservative than the second one in terms of having lower FDR at the expense of higher ADD. It is analytically shown that both procedures control the FDR under a specified tolerated level and are also scalable in the sense that they attain an ADD that does not increase asymptotically with the number of sensors. In addition, it is demonstrated that the proposed detection procedures are useful for trading off between reduced ADD and reduced average number of observations drawn until discovery. Numerical simulations are conducted for validating the analytical results and for demonstrating the properties of the proposed procedures.
\end{abstract}


\section{Introduction}\label{sec:intro}
Large-scale sensor networks are prominent new tools in various applications, e.g. Internet of Things (IoT), cyber-physical systems such as power grids, environmental monitoring, and wireless communication. These sensor networks can be used to perform statistical inference tasks \cite{POOR_ONE_SHOT,CHAUDHARI,CHEN_ZHANG_POOR_NON_BAYESIAN,CHEN_ZHANG_POOR_BAYESIAN}. An important statistical inference problem is sequential change-point detection \cite{PAGE,SHIRYAEV_OPTIMUM,LORDEN,POLLAK,MOUSTAKIDES,LAI,POOR_QUICKEST,TARTAKOVSKY_GENERAL,TARTAKOVSKY_ASIMPTOTIC} in which one is interested in detecting a rapid change in the underlying probability model, anomaly or adversarial activity as quickly as possible subject to a false positive constraint. Sensor networks, where each sensor observes a different data stream and communicates with a fusion center (FC) or cloud, can be deployed to detect multiple change points in a monitored environment.\\
\indent
Multiple change-point detection is closely related to multiple hypothesis testing. A widely-used performance criterion in multiple hypothesis testing is the false discovery rate (FDR), where the FDR is the expected proportion of the number of false discoveries among all discoveries \cite{BENJAMINI_HOCHBERG,EFRON_BOOK,FARCOMENI}. FDR control for multiple change-point detection has been considered in \cite{CHEN_ZHANG_POOR_NON_BAYESIAN} and \cite{CHEN_ZHANG_POOR_BAYESIAN,CHEN_ZHANG_POOR_FAULT} in the deterministic and Bayesian frameworks, respectively. These works assumed that {\em all} data streams are observed in parallel, which may not be feasible in large-scale sensor networks used in IoT. In the context of change-point detection, a Type I error (false positive) occurs if the detection procedure declares a change before the true change actually happens. In general, one would be interested in detecting the change point with minimum possible delay, while controlling the Type I error rate \cite{SHIRYAEV_OPTIMUM,TARTAKOVSKY_GENERAL}. In the Bayesian framework, the posterior probability of a change point having occurred, or some variation of it, is a commonly used test statistic \cite{TARTAKOVSKY_GENERAL,POOR_QUICKEST}.\\
\indent
Several works have considered discrete time single change-point detection in which only a part of the observations is available. In \cite{KUMAR_SLEEP}, Bayesian change-point detection was considered by monitoring only a minimal number of sensors at each time slot, where the change detection problem was modeled as a Markov decision process. A Bayesian method to minimize the average detection delay (ADD) subject to constraints on both the probability of false alarm and the observation cost was proposed in \cite{BANERJEE_VEERAVALLI}, where an on-off observation control policy was selected along with the stopping time at which the change is declared. Deterministic versions of this work were developed in \cite{BANERJEE_VEERAVALLI_MINIMAX,BANERJEE_VEERAVALLI_COMPOSITE,BANERJEE_VEERAVALLI_DECENTRALIZED} under different settings. Deterministic change-point detection for high-dimensional data with missing elements was considered in \cite{XIE_MISSING}. In \cite{GENG_LAI_NB} and \cite{GENG_BAYRAKTAR}, quickest change detection problems with sampling right constraints were considered in the deterministic and the Bayesian frameworks, respectively. Quickest deterministic change-point detection with observation scheduling was considered in \cite{REN_JOHANSSON_SCHEDULING}, where the decision maker chooses one of two different sequences of observations at each time slot. In \cite{REN_JOHANSSON_CENSORING}, deterministic change-point detection in sensor network with communication rate constraints was studied and adaptive censoring strategies were developed for the sensors. Quickest deterministic change-point detection over multiple data streams was considered in \cite{GENG_LAI_CONF}, where the observer can only observe one data stream at each time slot.\\
\indent
In this paper, we consider the problem of rapidly detecting change points in multiple data streams \cite{CHEN_ZHANG_POOR_NON_BAYESIAN,CHEN_ZHANG_POOR_BAYESIAN}. In particular, an FC receives statistically independent data streams from multiple sensors in a large-scale sensor network. Due to communication limitations, at a given time slot the FC monitors only a subset of the active data streams for which change points have not been declared yet. The subset size has a fixed proportion with respect to (w.r.t.) the number of active data streams. We assume that each data stream has an associated random change point.\\
\indent
The contributions of this paper are:
\begin{enumerate}
\item A Bayesian sequential procedure, named the sequential maximum a-posteriori probability (S-MAP) procedure, is proposed. This procedure detects the change points in all of the data streams, while controlling the FDR. The proposed procedure is based on sequentially updating the sensors' posterior probabilities of change points having occurred. Then, at each time slot we choose to monitor a subset of the sensors with the highest posterior probabilities within the allowed proportion. This approach aims to minimize the time between change point occurrence and its declaration by monitoring the sensors for which change-point occurrence is most probable given the data. The S-MAP procedure uses the same Type I error constraints as in \cite{CHEN_ZHANG_POOR_BAYESIAN} and extends this work to communication constrained scenarios. The FDR control of the S-MAP procedure is established using analytical tools.
\item We develop an improved S-MAP (IS-MAP) procedure that is less conservative than the S-MAP procedure in the sense that it has a lower ADD but higher FDR than the S-MAP procedure. The decrease of the ADD is obtained by reducing the detection threshold values of the IS-MAP procedure compared to the S-MAP procedure. It is proved analytically that the FDR of the IS-MAP procedure is still controlled under the desired level despite its lower detection threshold values. 
\item The asymptotic ADD behavior of the S-MAP and the IS-MAP procedures is established analytically for geometric prior distribution of the change points. It is shown that for any proportion value, both detection procedures are scalable in the sense that their asymptotic ADD does not increase with the number of data streams. In addition, the asymptotic ADD improvement that is obtained by using the IS-MAP procedure in comparison to the S-MAP procedure is characterized quantitatively.  
\item We conduct simulations in order to evaluate the performance and to verify the established theoretical properties of the S-MAP and the IS-MAP procedures. 
\item The S-MAP and the IS-MAP procedures are used for investigating the tradeoff between reducing the ADD and reducing the average number of observations (ANO) drawn until change points are declared. The proposed analysis can be useful for developing distributed statistical inference procedures using large-scale sensor networks in limited communication capability scenarios. 
\end{enumerate}

Preliminary results of this paper appear in our conference paper \cite{ICASSP2020} and a deterministic version of the proposed detection methods appears in \cite{CISS2020}. The remaining of this paper is organized as follows. In Section \ref{sec:Bayesian problem formulation}, we formulate the Bayesian multiple change-point detection problem. The S-MAP and the IS-MAP procedures are derived in Sections \ref{sec:S-MAP detection procedure} and \ref{sec:Improving the S-MAP procedure}, respectively, and their FDR control property is proved. Asymptotic ADD analysis of the S-MAP and the IS-MAP procedures is conducted in Section \ref{sec:ADD analysis of the S-MAP and the IS-MAP procedures}. Our simulations and conclusions appear in Sections \ref{sec:Numerical simulations} and \ref{sec:Conclusion}, respectively.

\section{Bayesian problem formulation}\label{sec:Bayesian problem formulation}
We consider $K$ statistically independent discrete time data streams denoted by $\{X_n^{(k)}\}_{n=1}^{\infty}$, $k\in[K]\define\{1,\ldots,K\}$. For the $k$th data stream there is a random change point, $t^{(k)}\geq 1,~\forall k\in[K]$, where the prior distribution of each change point is assumed to be known. Commonly, geometric distribution is assumed as a prior for discrete time change-point detection \cite{POOR_QUICKEST,MOULIN_BOOK}. The change points are assumed to be independent and identically distributed (i.i.d.) among the data streams. For the $k$th data stream, given its change point, $t^{(k)}$, we assume that $\{X_n^{(k)}\}_{n=1}^{t^{(k)}-1}$ are i.i.d. with known probability density $f_0$ and $\{X_n^{(k)}\}_{n=t^{(k)}}^{\infty}$ are i.i.d. with known probability density $f_1$. Due to communication limitations, at a given time slot we choose a subset of data streams to observe among the active data streams. Let $K_n\in\mathbb{N}$ denote the number of active sensors at time slot $n$. We set a fixed proportion value $q\in[0,1]$ and observe $\lceil qK_n \rceil\in\mathbb{N}$ of the active data streams, where $\lceil\cdot\rceil$ is the ceiling operator. The actual data vectors that are sequentially observed by the FC are denoted by $\{Y_n^{(s_n)}\}_{n=1}^{\infty}$, where $s_n\subset[K]$ is the subset of sensor indices that are monitored at time slot $n$. The filtration at time slot $n$ is the $\sigma$-algebra generated by the random vectors $Y_1^{s_1},\ldots,Y_n^{s_n}$, which is denoted by ${\mathcal{F}}_n\define\sigma(Y_1^{s_1},\ldots,Y_n^{s_n})$. In addition, we define the filtration of all the data as ${\mathcal{F}}_\infty\define\sigma(\{Y_m^{s_m}\}_{m=1}^\infty)$. For $k\in[K]$, the event $\{t^{(k)} \leq n \}$ stands for the case that change in the $k$th data stream has taken place before or at time slot $n$. We define the posterior probability of the event $\{t^{(k)} \leq n \}$ using the observations up to time slot $n$ as
\be\label{posterior_prob}
\pi_{n}^{(k)} \define P(t^{(k)}\leq n|{\mathcal{F}}_{n}),~n=1,2,\ldots, 
\ee 
and $\pi_{0}^{(k)}\define0$. We also define the likelihood ratio (LR), 
\be\label{LR_define}
L(X)\define\frac{f_1(X)}{f_0(X)}, 
\ee
and denote the Kullback-Leibler divergence of $f_1$ and $f_0$ as $D(f_1||f_0)$.\\
\indent
Under the assumption of i.i.d. change points, by using Bayes' rule we can recursively compute $\pi_{n}^{(k)}$ as follows:
\be\label{general_recursive_update}
\pi_{n}^{(k)} = \begin{cases}
\frac{L(X_n^{(k)})(\pi_{n-1}^{(k)} +\rho_n(1-\pi_{n-1}^{(k)}))}{L(X_n^{(k)})(\pi_{n-1}^{(k)} + \rho_n(1-\pi_{n-1}^{(k)})) + (1-\rho_n)(1-\pi_{n-1}^{(k)})},~ k\in s_n \\
\pi_{n-1}^{(k)} + \rho_n(1- \pi_{n-1}^{(k)}), ~~~~~~ \qquad \qquad \qquad   k \notin s_n
\end{cases}
\ee
$n\geq 1,~\forall k\in[K]$, where $\rho_n\define P(t^{(k)}=n|t^{(k)}\geq n)$ depends on the prior distribution of the change point. In \cite[Eq. (4.2)]{TARTAKOVSKY_GENERAL}, the statistic $\Lambda_n^{(k)}\define\frac{\pi_{n}^{(k)}}{1-\pi_{n}^{(k)}}$ is considered instead of $\pi_{n}^{(k)}$ and corresponding recursive update formula is presented for single change-point detection under general prior distribution. In case $k\in s_n$, then at time slot $n$ an observation is received from sensor $k$ and $\pi_{n}^{(k)}$ is computed using the observations received before time slot $n$, the prior distribution of $t^{(k)}$, and the new observation $X_n^{(k)}$. The posterior update for the case $k \notin s_n$ corresponds to the case in which at time slot $n$ we do not receive an observation from sensor $k$. In this case, $\pi_{n}^{(k)}$ is computed using only the observations received before time slot $n$ and the prior distribution of $t^{(k)}$. It is shown in \cite{GENG_BAYRAKTAR} that under mild conditions, the $k$th sensor posterior probability is a sufficient statistic for evaluating the $k$th stopping rule ADD and Type I error probability.\\
\indent
In the considered problem, we have to define multiple stopping rules $T^{(k)},~k\in[K]$, where the event $\{T^{(k)} \leq n \}$ is measurable w.r.t. $\mathcal{F}_n$. We define 
\be\label{FDR_define}
{\text{FDR}}\define {\rm{E}}\bigg[\frac{ V}{\max\{ R,1\}}\bigg], 
\ee
where ${\rm{E}}[\cdot]$ stands for the expectation. The term $V$ is the number of false discoveries, i.e. the size of the subset of $[K]$ s.t. $T^{(k)}<t^{(k)}$. The term $R$ denotes the number of change points declared, i.e. the size of the subset of $[K]$ s.t. $T^{(k)}<\infty$. We would like to control the FDR s.t. it will be no higher than a predefined tolerated level $\alpha\in(0,1)$. The ADD for the $k$th data stream is defined as 
\be\label{ADD_k_margin}
{\text{ADD}}_k\define{\rm{E}}[\max\{0,T^{(k)}-t^{(k)}\}].
\ee
Since we consider multiple statistically independent data streams, we define the overall ADD as
\be\label{ADD}
{\text{ADD}}\define \frac{1}{K}\sum_{k=1}^{K}{\text{ADD}}_k.
\ee
Assume that at time slot $n$, we have $K_n$ active data streams. Then, we observe $\lceil qK_n \rceil$ of them. We define
\be\label{ANO}
{\text{ANO}}\define {\rm{E}}\bigg[\frac{1}{K}\sum_{n=1}^{\underset{k\in[K]}{\sup}\{T^{(k)}\}}\left\lceil qK_n \right\rceil\bigg]. 
\ee
The ANO definition extends the definition from \cite{BANERJEE_VEERAVALLI}, which is defined for single change point detection, i.e. $K=1$. A difference between the definitions is that the ANO from \cite{BANERJEE_VEERAVALLI} does not consider the observations drawn after the change point occurs, while the ANO definition in \eqref{ANO} takes into account all the observations drawn until change points are declared. This is in order to properly evaluate the communication burden caused by transmissions of data streams from the sensors to the FC. In the following section, we propose the S-MAP procedure, which is a Bayesian multiple change-point detection procedure that controls the FDR under the limitation on the proportion of sensors communicating their data streams to the FC.

\section{S-MAP detection procedure}\label{sec:S-MAP detection procedure}
In this section, we propose a Bayesian detection procedure that is tasked to eventually discover all the random change points that occur in the monitored environment. At a given time slot, we consider each sensor individually and evaluate its posterior probability from \eqref{posterior_prob} using the recursive formula from \eqref{general_recursive_update}. At time slot $n$, we have $K_n$ active data streams of which we observe only a subset of size $\lceil qK_n \rceil\in\mathbb{N}$. The developed S-MAP procedure extends the method in \cite{CHEN_ZHANG_POOR_BAYESIAN} by proposing a rule for choosing the subset of $\lceil qK_n \rceil$ data streams to observe. In the S-MAP procedure, we use the posterior probability from \eqref{posterior_prob} as a test statistic, rather than the test statistic from \cite{TARTAKOVSKY_ASIMPTOTIC,CHEN_ZHANG_POOR_BAYESIAN}, which is based on a Bayesian version of the LR. The test statistic from \cite{TARTAKOVSKY_ASIMPTOTIC,CHEN_ZHANG_POOR_BAYESIAN} is used under a very strong global false alarm probability constraint \cite{TARTAKOVSKY_ASIMPTOTIC} that may be too conservative in terms of FDR control. Under the communication limitations, among the $K_n$ active data streams, we choose to observe the $\lceil qK_n \rceil$ data streams with the highest posterior probabilities of a change point having occurred. The motivation for the S-MAP approach is that we are interested in minimizing the time between the occurrence of a change point and its declaration using the sequentially updated posterior probabilities. The S-MAP procedure that monitors all of the active data streams, i.e. with $q=1$, is denoted as the parallel procedure. In the following, we describe the proposed S-MAP procedure.\\
\indent
We construct a descending set of $K$ thresholds $Q_r,~r\in[K]$, s.t. the detection on the $k$th data stream that samples until $\pi_{n}^{(k)}\geq Q_r$ has a Type I error probability that is smaller than or equal to $\frac{r\alpha}{K}$, where $\alpha\in(0,1)$ is the predefined FDR tolerance level. Formally,
\be\label{type1error}
P(\exists n < t^{(k)} \text{ s.t. } \pi^{(k)}_n \geq Q_r) \leq \frac{r}{K}\alpha.
\ee
According to \cite{TARTAKOVSKY_GENERAL} and \cite[p. 225]{MOULIN_BOOK}, the choice 
\be\label{high_thresholds}
Q_r=1-\frac{r\alpha}{K} 
\ee
ensures that \eqref{type1error} is satisfied. The proposed detection procedure is divided into sampling stages. Each sampling stage may take several time slots. In the beginning of a sampling stage, we gather all the active data streams and obtain observations from a subset of them, according to the S-MAP approach. This process is repeated at each time slot sequentially, until at least one active data stream posterior probability exceeds its corresponding threshold. Then, we declare changes for some of the active data streams, which are then eliminated from the active data streams set.\\
\indent
Let $I_j$ denote the set of indices of active data streams with cardinality $|I_j|$ at the beginning of the $j$th sampling stage and let $n_j$ denote the time slot at the end of the $j$th sampling stage. Note that $I_1=[K]$ and $n_0=0$. The $j$th stage of sampling is described as follows: 
\begin{enumerate}
\item Sample the $\lceil q|I_j| \rceil$ data streams with the currently highest posterior probabilities. 
\item Update the posterior probabilities of the sensors with active data streams using \eqref{general_recursive_update}. 
\item Sort the updated posterior probabilities in ascending order as $\pi_{n}^{(i(n,l))}$, where $i(n,l)$ denotes the index of the $l$th ordered posterior probability at time slot $n$. 
\item Repeat this process until time slot $n_j$ in which at least one of the posterior probabilities is higher than its corresponding threshold, i.e. $n_j=\min\{n>n_{j-1}:\exists l\in[|I_j|], \pi_{n}^{(i(n,l))}\geq Q_{K-l+1}\}$.
\item Declare change points for the data streams $i(n_j,l_j),i(n_j,l_j+1),\ldots,i(n_j,|I_j|)$, where $l_j=\min\{l\in[|I_j|]:\pi_{n_j}^{(i(n_j,l))}\geq Q_{K-l+1}\}$ and remove these data streams from the set of active data streams. 
\item Update $I_{j+1}$ to be the set of indices of the remaining active data streams. Stop the procedure if $|I_{j+1}|=0$.
\end{enumerate}
In the following theorem, we show that we control the FDR of the S-MAP procedure to remain under the upper bound constraint $\alpha\in(0,1)$.
\begin{Theorem}\label{T_FDR_S_MAP}
For upper bound constraint $\alpha\in(0,1)$, the S-MAP procedure satisfies
\be\label{FDR_defineS_MAP}
{\text{FDR}}\leq\alpha. 
\ee 
\end{Theorem}
\begin{proof}
Recall that we choose the thresholds $Q_r,~r\in[K]$ from \eqref{high_thresholds}, s.t. \eqref{type1error} is satisfied. Thus, by following the lines of the FDR control proofs in \cite{CHEN_ZHANG_POOR_NON_BAYESIAN,CHEN_ZHANG_POOR_BAYESIAN}, we obtain that the FDR is controlled by the proposed S-MAP procedure under the upper bound constraint $\alpha$.
\end{proof}
In the following section, we propose an alternative detection procedure that is less conservative than the S-MAP procedure in terms of FDR control. Therefore, the proposed alternative procedure has improved performance in terms of ADD and ANO compared to the S-MAP procedure.

\section{Improving the S-MAP procedure}\label{sec:Improving the S-MAP procedure}
In order to guarantee FDR control, the S-MAP procedure uses the false alarm constraints from \eqref{type1error}, which are the same false alarm constraints as in \cite{CHEN_ZHANG_POOR_BAYESIAN} to guarantee FDR control. However, we show in this section that the false alarm constraints from \cite{CHEN_ZHANG_POOR_BAYESIAN} may be too conservative and the corresponding posterior probability threshold values may be too high. We propose the IS-MAP detection procedure, which is similar to the S-MAP procedure except that its threshold values are lower than the thresholds of the S-MAP procedure. Since the IS-MAP procedure uses lower threshold values, then for a fixed proportion, $q$, the ADD and ANO will decrease compared to the S-MAP procedure, i.e. the ADD and ANO performance will improve. Moreover, using the lower thresholds, we prove that we can still control the FDR under the desired level, $\alpha$. In the IS-MAP procedure, we construct a set of $K$ thresholds $Q_k,~k\in[K]$, s.t. the detection on the $k$th data stream that samples until $\pi_{n}^{(k)}\geq Q_k$ has an individual Type I error probability that is smaller than or equal to $\alpha$, where $\alpha\in(0,1)$ is the predefined FDR tolerated level. Formally,
\be\label{type1error_alter}
P(\exists n < t^{(k)} \text{ s.t. } \pi^{(k)}_n \geq Q_k) \leq \alpha.
\ee
According to \cite{TARTAKOVSKY_GENERAL} and \cite[p. 225]{MOULIN_BOOK}, the choice 
\be\label{thresholds_eq_alpha}
Q_{k}=Q= 1-\alpha,~\forall k\in[K],
\ee
ensures that \eqref{type1error_alter} is satisfied. Since the thresholds of the IS-MAP procedure are all equal to $Q$, its $j$th sampling stage can be written in a more compact form than the corresponding sampling stage of the S-MAP procedure. Let $I_j$ denote the set of indices of active data streams with cardinality $|I_j|$ at the beginning of the $j$th sampling stage and let $n_j$ denote the time slot at the end of the $j$th sampling stage. The $j$th stage of sampling is described as follows:
\begin{enumerate}
\item Sample the $\lceil q|I_j| \rceil$ data streams with highest posterior probabilities. 
\item Update the posterior probabilities of the sensors with active data streams using \eqref{general_recursive_update}. 
\item Repeat this process until time slot $n_j$ in which at least one of the posterior probabilities is higher than the threshold $Q$, i.e. $n_j=\min\{n>n_{j-1}:\exists k\in I_j, \pi_{n}^{(k)}\geq Q\}$.
\item Declare change points for all the data streams with indices in $I_j$ whose posterior probabilities are higher than or equal to $Q$ and remove these data streams from the set of active data streams. 
\item Update $I_{j+1}$ to be the set of indices of the remaining active data streams. Stop the procedure if $|I_{j+1}|=0$.
\end{enumerate}
In the following theorem, we show that the FDR of the IS-MAP procedure satisfies the desired upper bound constraint.
\begin{Theorem}\label{T_FDR_IS_MAP}
For upper bound constraint $\alpha\in(0,1)$, the IS-MAP procedure satisfies
\be\label{FDR_defineIS_MAP}
{\text{FDR}}\leq\alpha. 
\ee 
\end{Theorem}
\begin{proof}
The proof is given in Appendix \ref{App_T_FDR_IS_MAP}.
\end{proof}
As mentioned previously, the proposed IS-MAP procedure is similar to the S-MAP procedure from Section \ref{sec:S-MAP detection procedure}, except that the procedures use different thresholds in order to guarantee the FDR control. Since the thresholds of the IS-MAP procedure in \eqref{thresholds_eq_alpha} are smaller than the thresholds of the S-MAP procedure, then for a fixed proportion, $q$, the IS-MAP procedure will have a lower ADD and ANO than the S-MAP procedure, while the FDR of the IS-MAP procedure will be higher than the S-MAP FDR. It should be noted that in case of model uncertainty, FDR control is not guaranteed for the S-MAP and the IS-MAP procedures. Then, depending on the application, if ADD and ANO are more significant than FDR, the IS-MAP procedure should be implemented rather than the S-MAP procedure, while if FDR is more significant than ADD and ANO, then the S-MAP procedure may be preferred. In the following section, we analyze the asymptotic ADD behavior of the S-MAP and the IS-MAP procedures under the assumption of geometric prior distribution for the change points.

\section{ADD analysis of the S-MAP and the IS-MAP procedures}\label{sec:ADD analysis of the S-MAP and the IS-MAP procedures}
In this section, we derive asymptotic lower and upper bounds on the ADD of the S-MAP and the IS-MAP procedures for $\alpha\to 0$ and a fixed number of data streams $K$. Then, we characterize the behavior of these bounds as $K\to\infty$. For simplicity of the analysis, we assume that the prior distribution of each change point obeys a geometric distribution with common parameter $\rho\in(0,1)$, i.e. 
\be\label{geom_prior}
P(t^{(k)}=m)=\rho(1-\rho)^{m-1},~\forall m=1,2,\ldots,\forall k\in[K]. 
\ee
The geometric prior distribution is commonly assumed in change-point detection problems. This is a memoryless distribution that is both mathematically
convenient and provides a reasonable model in practical applications \cite{POOR_QUICKEST,MOULIN_BOOK}. Under the assumption of i.i.d. change points with geometric priors, it is shown in \cite{BANERJEE_VEERAVALLI,GENG_BAYRAKTAR} that the posterior probability of the $k$th sensor evolves in a sequential manner via the recursion
\be\label{recursive_update}
\pi_{n}^{(k)} = \begin{cases}
\frac{L(X_n^{(k)})(\pi_{n-1}^{(k)} + \rho(1-\pi_{n-1}^{(k)}))}{L(X_n^{(k)})(\pi_{n-1}^{(k)} + \rho(1-\pi_{n-1}^{(k)})) + (1-\rho)(1-\pi_{n-1}^{(k)})},~ k\in s_n \\
\pi_{n-1}^{(k)} + \rho(1- \pi_{n-1}^{(k)}), ~~~~~~ \qquad \qquad \qquad   k \notin s_n
\end{cases}
\ee
$n\geq 1,~\forall k\in[K]$. It can be seen that the recursive formula in \eqref{recursive_update} is obtained by substituting $\rho_n=\rho$ in \eqref{general_recursive_update}.\\
\indent
Under communication limitations, the FC observes a subsequence of the complete observation sequence from each sensor. According to the maximum a-posteriori probability (MAP) approach, the indices of the monitored observations are random and determined online based on the proportion, $q$, and the posterior probability values of the active sensors at each time slot. Therefore, it is difficult to characterize the subsequence of observations obtained from each sensor. In order to obtain asymptotic bounds on the ADD of the S-MAP and the IS-MAP procedures, we begin by considering a single change-point detection with the posterior update from \eqref{recursive_update}. Thus, we consider the observation sequence $\{X_n\}_{n=1}^{\infty}$ with change point $t$ and stopping rule of the form 
\be\label{single_stop_rule}
T=\inf\{n\in\mathbb{N}:\pi_{n}\geq 1-\eta\}, \eta\in(0,1). 
\ee
We assume that only a {\em{subsequence}} of the complete observation sequence is obtained. It is shown in \cite{GENG_BAYRAKTAR} that for any subsequence of observations, the ADD of the stopping rule in the form of \eqref{single_stop_rule} as $\eta\to 0$ satisfies 
\be\label{single_lower_bound}
{\text{ADD}}\geq\frac{|\log\eta|}{D(f_1||f_0)+|\log(1-\rho)|}(1+o_\eta(1))
\ee
and
\be\label{single_upper_bound}
{\text{ADD}}\leq\frac{|\log\eta|}{|\log(1-\rho)|}(1+o_\eta(1)),
\ee
where $o_\eta(1)\to 0$ as $\eta\to 0$. The asymptotic ADD lower bound from \eqref{single_lower_bound} is attained when the complete observation sequence is available. The asymptotic ADD upper bound from \eqref{single_upper_bound} is attained when we do not take observations and the stopping rule is based only on the prior.\\
\indent
In the following theorem, using \eqref{single_lower_bound}-\eqref{single_upper_bound}  we derive asymptotic lower and upper bounds on the ADDs of the S-MAP and the IS-MAP procedures as $\alpha\to 0$. These ADD bounds do not require any assumptions on the subsequence of observations obtained from each sensor.
\begin{Theorem}\label{bounds_ADD}
For $\alpha\to 0$ and any proportion of observed sensors, $q$, we obtain
\be\label{asympt_lower_bound_S_MAP}
{\text{ADD}}_{\text{S-MAP}}\geq\frac{|\log\alpha|}{D(f_1||f_0)+|\log(1-\rho)|}(1+o_\alpha(1)),
\ee
\be\label{asympt_upper_bound_S_MAP}
{\text{ADD}}_{\text{S-MAP}}\leq\frac{\log K-\frac{1}{K}\log K!+|\log\alpha|}{|\log(1-\rho)|}(1+o_\alpha(1)),
\ee
\be\label{asympt_lower_bound_IS_MAP}
{\text{ADD}}_{\text{IS-MAP}}\geq\frac{|\log\alpha|}{D(f_1||f_0)+|\log(1-\rho)|}(1+o_\alpha(1)),
\ee
and
\be\label{asympt_upper_bound_IS_MAP}
{\text{ADD}}_{\text{IS-MAP}}\leq\frac{|\log\alpha|}{|\log(1-\rho)|}(1+o_\alpha(1)).
\ee
\end{Theorem}
\begin{proof}
The proof is given in Appendix \ref{App_bounds_ADD}.
\end{proof}
For the ADD of the stopping rule, $T$, from \eqref{single_stop_rule} we can derive a tighter upper bound than \eqref{single_upper_bound} under some assumptions on the subsequence of observations obtained for the detection. Let us denote by $\{X_{V_n}\}_{n=1}^{\infty}$ the subsequence of the complete observation sequence, where $V_0\define 0$ and $V_1,V_2,\ldots$ are the integer time slots in which observations are obtained for the detection of the single change point, $t$, using the stopping rule, $T$. Equivalently, we sample the complete observation sequence with intervals 
\be\label{interval_define}
\zeta_n\define V_n-V_{n-1}\geq 1,n\in\mathbb{N}.
\ee
In addition, we define 
\be\label{arithmetic_mean}
\zeta^{(N)}\define\frac{1}{N}\sum_{n=1}^{N}\zeta_n=\frac{V_N}{N},
\ee
which is the average length of intervals in which we sample $N$ observations from the observation sequence, the stopping rule,
\be\label{single_stop_rule_ignore}
\Gamma\define\inf\{n\in\mathbb{N}:\pi_{V_n}\geq 1-\eta\},
\ee
and the random change point,
\be\label{single_change_point_ignore}
\gamma\define\inf\{n\in\mathbb{N}:V_n\geq t\}.
\ee
The stopping rule and change point from \eqref{single_stop_rule_ignore} and \eqref{single_change_point_ignore}, respectively, represent the case in which we only count time slots where observations are obtained. The time slots, $\{V_n\}_{n=1}^{\infty}$, and intervals, $\{\zeta_n\}_{n=1}^{\infty}$, may be unknown. For the derivation of a tighter asymptotic upper bound on the ADD of the stopping rule, $T$, we only assume that the intervals are bounded, i.e. there exists $1\leq{\mathcal{B}}<\infty$ s.t.
\be\label{bounded_interval_define}
\zeta_n\leq{\mathcal{B}},~\forall n\in\mathbb{N},
\ee
there exists $\zeta\in[1,{\mathcal{B}}]$ s.t.
\be\label{convergent_interval_mean}
\underset{N\to\infty}{\lim}\zeta^{(N)}=\zeta,
\ee
and
\be\label{convergent_expectation_mean}
{\rm{E}}[\zeta^{(\Gamma)}\max\{0,\Gamma-\gamma\}]=\zeta{\rm{E}}[\max\{0,\Gamma-\gamma\}](1+o_\eta(1)).
\ee
From \eqref{interval_define}-\eqref{arithmetic_mean}, $\zeta^{(\Gamma)}=\frac{V_\Gamma}{\Gamma}$, $\zeta^{(\gamma)}=\frac{V_\gamma}{\gamma}$, $\zeta_\gamma\define V_\gamma-V_{\gamma-1}$, and the specific value of $\zeta$ may be unknown. The assumption in \eqref{convergent_expectation_mean} essentially requires that $\Gamma\to\infty$ as $\eta\to0$. In the following proposition, we derive an asymptotic ADD upper bound for the stopping rule, $T$, which is tighter than \eqref{single_upper_bound}.
\begin{proposition}\label{single_ADD}
Assume that \eqref{bounded_interval_define}-\eqref{convergent_expectation_mean} are satisfied. Then, as $\eta\to 0$ the ADD of the stopping rule $T$ from \eqref{single_stop_rule} satisfies 
\be\label{single_upper_bound_assumption}
{\text{ADD}}\leq\frac{|\log\eta|}{\frac{1}{\zeta}D(f_1||f_0)+|\log(1-\rho)|}(1+o_\eta(1)).
\ee
\end{proposition}
\begin{proof}
The proof is given in Appendix \ref{App_single_ADD}.	
\end{proof}
\noindent
It should be noted that a special case of \eqref{single_upper_bound_assumption} with $\zeta_n=\zeta<\infty,n\in\mathbb{N}$, was proved in \cite{GENG_BAYRAKTAR}.\\
\indent 
Assume that each stopping rule in the S-MAP procedure satisfies the ADD upper bound in \eqref{single_upper_bound_assumption} with $\zeta=g_k^{\text{S-MAP}}$ and that each stopping rule in the IS-MAP procedure satisfies the ADD upper bound in \eqref{single_upper_bound_assumption} with $\zeta=g_k^{\text{IS-MAP}}$, $\forall k\in[K]$. In addition, assume that $g^*\define{\sup}\{\{g_k^{\text{S-MAP}}\}_{k=1}^K,\{g_k^{\text{IS-MAP}}\}_{k=1}^K\}<\infty$. Then, in a similar manner to the derivation of the upper bounds in \eqref{asympt_upper_bound_S_MAP} and \eqref{asympt_upper_bound_IS_MAP}, we obtain tighter asymptotic ADD upper bounds for the S-MAP and the IS-MAP procedures, given by
\be\label{asympt_upper_bound_S_MAP_assumption}
{\text{ADD}}_{\text{S-MAP}}\leq\frac{\log K-\frac{1}{K}\log K!+|\log\alpha|}{\frac{1}{g^*}D(f_1||f_0)+|\log(1-\rho)|}(1+o_\alpha(1))
\ee
and
\be\label{asympt_upper_bound_IS_MAP_assumption}
{\text{ADD}}_{\text{IS-MAP}}\leq\frac{|\log\alpha|}{\frac{1}{g^*}D(f_1||f_0)+|\log(1-\rho)|}(1+o_\alpha(1)),
\ee
respectively.\\
\indent
In \eqref{asympt_lower_bound_S_MAP}, \eqref{asympt_upper_bound_S_MAP}, and \eqref{asympt_upper_bound_S_MAP_assumption} and in \eqref{asympt_lower_bound_IS_MAP}, \eqref{asympt_upper_bound_IS_MAP}, and \eqref{asympt_upper_bound_IS_MAP_assumption}, we obtained asymptotic ADD bounds for the S-MAP and the IS-MAP procedures, respectively. For any fixed proportion, $q$, of observed data streams and for sufficiently small $\alpha\neq 0$ these bounds hold. We characterize the behavior of these bounds as $K$ increases towards $\infty$ in order to investigate the scalability of the S-MAP and the IS-MAP procedures, as the number of data streams increases. Let
\be\label{ADD_LB_joint}
{\text{ADD}}^{(\alpha)}_{\text{LB}}\define\frac{|\log\alpha|}{D(f_1||f_0)+|\log(1-\rho)|}
\ee
denote the asymptotic ADD lower bound for both the S-MAP and the IS-MAP procedures. It can be seen that this lower bound is a finite constant w.r.t. $K$.\\
\indent 
We denote the asymptotic ADD upper bounds for the S-MAP procedure as
\be\label{ADD_UB_S_MAP}
{\text{ADD}}^{(\alpha,K)}_{\text{S-MAP,UB}}\define\frac{\log K-\frac{1}{K}\log K!+|\log\alpha|}{|\log(1-\rho)|}.
\ee
and
\be\label{ADD_UB_S_MAP_assume}
{\text{ADD}}^{(\alpha,g^*,K)}_{\text{S-MAP,UB}}\define\frac{\log K-\frac{1}{K}\log K!+|\log\alpha|}{\frac{1}{g^*}D(f_1||f_0)+|\log(1-\rho)|}.
\ee
Consider the sequence $\{\log K-\frac{1}{K}\log K!\}_{K=1}^\infty$. Using \cite[Eq. (5)]{SANDOR} and Stirling's approximation (see e.g. \cite{SANDOR,ROBINS}) and applying some algebraic manipulations, it can be verified that this sequence is monotonically increasing and converges to $1$. Thus, we obtain that ${\text{ADD}}^{(\alpha,K)}_{\text{S-MAP,UB}}$ and ${\text{ADD}}^{(\alpha,g^*,K)}_{\text{S-MAP,UB}}$ are monotonically increasing with $K$ and converge to a finite constant, i.e.
\be\label{ADD_UB_S_MAP_lim}
\underset{K\to\infty}{\lim}{\text{ADD}}^{(\alpha,K)}_{\text{S-MAP,UB}}=\frac{1+|\log\alpha|}{|\log(1-\rho)|}
\ee
and
\be\label{ADD_UB_S_MAP_assume_lim}
\underset{K\to\infty}{\lim}{\text{ADD}}^{(\alpha,g^*,K)}_{\text{S-MAP,UB}}=\frac{1+|\log\alpha|}{\frac{1}{g^*}D(f_1||f_0)+|\log(1-\rho)|}.
\ee
In a similar manner to \eqref{ADD_UB_S_MAP}-\eqref{ADD_UB_S_MAP_assume}, we denote
\be\label{ADD_UB_IS_MAP}
{\text{ADD}}^{(\alpha)}_{\text{IS-MAP,UB}}\define\frac{|\log\alpha|}{|\log(1-\rho)|}
\ee
and
\be\label{ADD_UB_IS_MAP_assume}
{\text{ADD}}^{(\alpha,g^*)}_{\text{IS-MAP,UB}}\define\frac{|\log\alpha|}{\frac{1}{g^*}D(f_1||f_0)+|\log(1-\rho)|}.
\ee
The upper bounds in \eqref{ADD_UB_IS_MAP}-\eqref{ADD_UB_IS_MAP_assume} are finite constants w.r.t. $K$.\\
\indent
The sequence $\{\log K-\frac{1}{K}\log K!\}_{K=1}^\infty$ is nonnegative and thus,
\be\label{ADD_finite_K_comparison}
{\text{ADD}}^{(\alpha)}_{\text{IS-MAP,UB}}\leq{\text{ADD}}^{(\alpha,K)}_{\text{S-MAP,UB}}
\ee
and
\be\label{ADD_finite_K_comparison_assume}
{\text{ADD}}^{(\alpha,g^*)}_{\text{IS-MAP,UB}}\leq{\text{ADD}}^{(\alpha,g^*,K)}_{\text{S-MAP,UB}}.
\ee
In addition, by comparing \eqref{ADD_UB_IS_MAP}-\eqref{ADD_UB_IS_MAP_assume} to \eqref{ADD_UB_S_MAP}-\eqref{ADD_UB_S_MAP_assume} as $K\to\infty$, we obtain
\be\label{ADD_K_comparison}
\begin{split}
\underset{K\to\infty}{\lim}\frac{{\text{ADD}}^{(\alpha)}_{\text{IS-MAP,UB}}}{{\text{ADD}}^{(\alpha,K)}_{\text{S-MAP,UB}}}&=\underset{K\to\infty}{\lim}\frac{{\text{ADD}}^{(\alpha,g^*)}_{\text{IS-MAP,UB}}}{{\text{ADD}}^{(\alpha,g^*,K)}_{\text{S-MAP,UB}}}\\
&=\frac{|\log\alpha|}{1+|\log\alpha|}<1,
\end{split}
\ee
where the second equality is obtained by substituting \eqref{ADD_UB_S_MAP_lim}-\eqref{ADD_UB_IS_MAP_assume}. The results in \eqref{ADD_finite_K_comparison}-\eqref{ADD_K_comparison} demonstrate the ADD improvement obtained by using the IS-MAP procedure instead of the S-MAP procedure.\\
\indent 
The presented asymptotic ADD results hold for any proportion value, $q$. However, it is expected that the S-MAP ADD and the IS-MAP ADD will increase as the proportion of monitored sensors decreases. An intuitive explanation for this phenomenon is as follows: For fixed $\pi_{n-1}^{(k)}$, the posterior probability in \eqref{recursive_update} is monotonically nondecreasing with the LR, $L(\cdot)$. After a change occurs, we receive samples from $f_1$. By taking the expectation of the difference $L(X)-1$ w.r.t. $f_1$ and using ${\rm{E}}_{f_1}[L(X)]={\rm{E}}_{f_0}[L^2(X)]$ and ${\rm{E}}_{f_0}[L(X)]=1$, we obtain 
\be\label{LR_monotone_increase}
{\rm{E}}_{f_1}[L(X)-1]={\rm{E}}_{f_0}[(L(X)-1)^2]\geq 0. 
\ee
The case $L(X)=1$ corresponds to the case in which we choose not to monitor the corresponding sensor. Thus, as the number of observations increases, the threshold will usually be exceeded in an earlier time slot and consequently, the ADD will usually be lower. An advantage of observing only a small subset of sensors is that the ANO for the detection task may decrease, which reduces the communication burden. Consequently, we identify a tradeoff between the ADD and the ANO. We will investigate this tradeoff in Section \ref{sec:Numerical simulations}.

\section{Numerical simulations}\label{sec:Numerical simulations}
In this section, we evaluate the performance of the proposed S-MAP and IS-MAP procedures in terms of FDR, ADD, and ANO. In addition, the analytical results from Sections \ref{sec:S-MAP detection procedure}-\ref{sec:ADD analysis of the S-MAP and the IS-MAP procedures} are verified in the simulations. The simulation results are based on $1000$ Monte Carlo runs. We generate the true change points independently for each sensor from a geometric distribution with parameter $\rho=0.01$ and assume that we know this parameter when applying the procedure. It should be noted that in case $\rho$ is unknown then by assuming a sufficiently low value for $\rho$, the FDR of the S-MAP and the IS-MAP procedures may still be controlled under the desired upper bound. The reason is that the posterior probabilities from \eqref{recursive_update} decrease as $\rho$ decreases. If the assumed value of $\rho$ is lower than the true value of $\rho$, the change-points will usually be declared in later time slots than in the case in which the true value of $\rho$ is used. Thus, the FDR will not increase. In all cases, we set the FDR upper bound as $\alpha = 0.1$.\\
\indent
For comparison purposes, we implement and evaluate the performance of two additional procedures. The first procedure is a simplified version of the S-MAP procedure, which is referred to as the simple procedure. This procedure simplifies S-MAP from Section \ref{sec:S-MAP detection procedure} by replacing the method of choosing the subset of sensors to monitor. In the simple procedure, at each time slot we randomly choose a subset of active sensors with consecutive indices to monitor within the allowed proportion. Following the FDR control proofs in \cite{CHEN_ZHANG_POOR_NON_BAYESIAN,CHEN_ZHANG_POOR_BAYESIAN}, it can be shown that the simple procedure controls the FDR under the predefined upper bound. This procedure is implemented in order to verify that the MAP approach for choosing the subset of sensors to monitor, as used in the S-MAP procedure, improves the ADD performance compared to randomly choosing this subset, as used in the simple procedure. The second method implemented for comparisons is the fully parallel procedure of \cite{CHEN_ZHANG_POOR_BAYESIAN}, named D-FDR, that observes all the data streams. The FDR control of the D-FDR procedure is shown in \cite{CHEN_ZHANG_POOR_BAYESIAN}. In this procedure, the following test statistic is used
\be\label{ALR_define}
G_{n}^{(k)} = \sum_{m=1}^{\infty}P(t^{(k)}=m)\prod_{i=m}^{n}L(X_i^{(k)}),~n=1,2,\ldots.
\ee
This test statistic is the average LR (ALR) between the hypotheses that the change occurs at $t^{(k)}=m<\infty$ and that the change never occurs, $t^{(k)}=\infty$. This ALR test statistic is recursively updated according to the following formula:
\be\label{ALR_update}
G_{n}^{(k)} = G_{n-1}^{(k)}L(X_n^{(k)})+P(t^{(k)}\geq n+1)(1-L(X_n^{(k)})),
\ee
where $G_{0}^{(k)}\define 1$. For $q=1$, the D-FDR procedure is similar to the S-MAP procedure except that it uses the ALR test statistic, rather than the posterior probability test statistic, with the thresholds
\be\label{high_ALR_thresholds}
Q_r=\frac{K}{r\alpha},~r\in[K],
\ee
in order to guarantee the same false positive constraints as in \eqref{type1error}. Assume that for the $k$th data stream, the corresponding threshold is $Q_{r_k}=\frac{K}{r_k\alpha}$, $r_k\in[K]$. It is shown in \cite{TARTAKOVSKY_ASIMPTOTIC} that in this case, using the ALR test statistic with the threshold $Q_{r_k}$ is equivalent to using the posterior probability test statistic with the threshold
\be\label{high_posterior_thresholds}
Q_{r_k}^*=1-p(t^{(k)}\geq n+1)\frac{r_k\alpha}{K}.
\ee
Thus, from \eqref{high_thresholds}, \eqref{thresholds_eq_alpha}, and \eqref{high_posterior_thresholds}, the posterior probability thresholds of the D-FDR procedure are higher than the posterior probability thresholds of the S-MAP and the IS-MAP procedures. Consequently for $q=1$, the ADD and ANO of the S-MAP and the IS-MAP procedures will be lower than the ADD and ANO of the D-FDR procedure.\\
\indent
In Subsection \ref{subsec:Gaussian distribution scenario}, we consider multiple change-point detection with known Gaussian distributions and in Subsection \ref{subsec:General model with uncertainty and known p-values}, we consider a general model under some uncertainty and use $p$-values \cite{BENJAMINI_HOCHBERG,EFRON_BOOK,BAILEY,EFRON_LARGE,POUNDS,HALME} from each sensor as observations for the multiple change-points detection. It should be noted that in the simulations, we assume that we have a sufficient number of observations for declaring the changes so there are no Type II errors corresponding to infinite ADD.

\subsection{Gaussian distribution scenario}\label{subsec:Gaussian distribution scenario}
We consider Gaussian distributions with a change in the mean and set $f_0 = \mathcal{N}(0,1)$ and $f_1 = \mathcal{N}(1,1)$ as depicted in Fig. \ref{pdf_Gaussian}. First, for $K=10,100,200,500,1000$, we examine the FDR control of the proposed S-MAP and IS-MAP procedures with $\{q=0.05m\}_{m=1}^{20}$, where $q\in[0,1]$ is the proportion of monitored sensors. The proportion $q=1$ corresponds to the parallel versions of the S-MAP and the IS-MAP procedures that observe all the active data streams at each time slot. Due to space limitations, we do not present tables of all the estimated FDR results. The resulting minimum and maximum estimated FDR values of the S-MAP procedure are $0.028$ and $0.037$, respectively, while the resulting minimum and maximum estimated FDR values of the IS-MAP procedure are $0.058$ and $0.068$, respectively. Consequently, both procedures control the FDR under the upper bound $\alpha=0.1$. These results confirm the analytical results in Theorems \ref{T_FDR_S_MAP} and \ref{T_FDR_IS_MAP}. The S-MAP FDR values are lower than the IS-MAP FDR values, since the S-MAP procedure is more conservative and uses higher thresholds than the IS-MAP procedure. For both the S-MAP and the IS-MAP procedures there is still a gap between the FDR values and the upper bound $\alpha$. This result follows from the choices of thresholds in \eqref{high_thresholds} and \eqref{thresholds_eq_alpha} for the S-MAP and the IS-MAP procedures, respectively, that neglect the overshoot in the stopping rule \cite{TARTAKOVSKY_GENERAL}.\\
\indent
In Fig. \ref{ADD_Gaussian}, we evaluate the ADD of the procedures: D-FDR, S-MAP with $q=0.5,1$, simple procedure with $q=0.5$, and IS-MAP with $q=0.5,1$ versus $K=10,100,200,500,1000$. It can be seen that all the considered procedures have an approximately constant ADD as $K$ increases, which verifies the analytical results in Section \ref{sec:ADD analysis of the S-MAP and the IS-MAP procedures}. The parallel version of the IS-MAP procedure, i.e. for $q=1$, has the lowest ADD. Moreover, it can be seen that the IS-MAP procedure with $q=0.5$ outperforms the parallel version of the S-MAP procedure and the D-FDR procedure. These results demonstrate the advantage of using the IS-MAP procedure instead of the S-MAP or the D-FDR procedures in terms of ADD. The simple procedure with $q=0.5$ has the highest ADD among the considered procedures implying that the proposed MAP approach is desirable for choosing the sensors to monitor at each time slot within the allowed proportion. In Fig. \ref{ANO_Gaussian}, we evaluate the ANO versus $K$ of the procedures: D-FDR, S-MAP with $q=0.5,1$, and IS-MAP with $q=0.5,1$. It can be seen that IS-MAP with $q=0.5$ has the lowest and the D-FDR has the highest ANO. In addition, it can be seen that for all the procedures, the ANO is approximately a constant w.r.t. $K$.\\
\indent
In the upper and middle plots of Fig. \ref{ADD_ANO_prop_tradeoff_Gaussian}, we plot the ADDs and ANOs, respectively, of the S-MAP and the IS-MAP procedures for $K=1000$ versus the proportion values $\{q=0.05m\}_{m=1}^{20}$. It can be seen that for any of the considered proportions, the IS-MAP procedure has lower ADD and ANO than the S-MAP procedure. In addition, for both procedures the ADD decreases as the proportion increases, while the ANO increases approximately linearly as the proportion increases. Thus, we notice a tradeoff between ADD and ANO as we change the proportion value, $q$. It can be seen that for both procedures there is no significant increase in ADD when the proportion decreases from $q=1$ to $q=0.3$, whereas the ANO increases significantly as we increase $q$ towards $1$. This result implies that in this example it may be a waist of resources to monitor all the active data streams in parallel. In the lower plot of Fig. \ref{ADD_ANO_prop_tradeoff_Gaussian}, we plot a curve connecting the ADD-ANO points of the S-MAP and the IS-MAP procedures from the upper and middle plots of Fig. \ref{ADD_ANO_prop_tradeoff_Gaussian}. It can be seen that in this example there is a clear tradeoff between the ADD and ANO, i.e. as the proportion, $q$, increases the ADD becomes lower, while the ANO becomes higher.\\
\indent
In order to evaluate the performance of the procedures using both the ADD and the ANO as criteria, we define a weighted risk, 
\be\label{weighted_risk}
(1-c){\text{ADD}}+c{\text{ANO}}, 
\ee
where $c\in[0,1]$ sets the weighting between the ADD and the ANO. For $c=0$ we are only interested in the ADD, while for $c=1$ we are only interested in the ANO. In the upper plot of Fig. \ref{weighted_optimal_risk_Gaussian}, we compare the weighted risks of the S-MAP and the IS-MAP procedures with different proportions $\{q=0.05m\}_{m=1}^{20}$ versus the proportion size for $c=0.2$. It can be seen that the weighted risk of the IS-MAP procedure is lower than the weighted risk of the S-MAP procedure. For both the S-MAP and the IS-MAP procedures, the best tradeoff among the considered proportions is achieved with the proportion $q=0.3$. Thus, when both the ADD and the ANO are taken into account it may not be necessary to monitor all the active data streams in parallel, i.e. to choose $q=1$.\\
\indent
In the lower plot of Fig. \ref{weighted_optimal_risk_Gaussian}, for both the S-MAP and the IS-MAP procedures, we present the best proportion among the proportions $\{q=0.05m\}_{m=1}^{20}$ in terms of the weighted risk in \eqref{weighted_risk}, i.e. the proportion with lowest risk, versus the weighting coefficient $c$. It can be seen that for both procedures, as $c$ increases the best proportion does not increase. Moreover, in most of the considered cases the best proportion decreases as $c$ increases. Thus, as we put a higher weight on the ANO compared to the ADD we should usually choose a lower proportion of data streams to observe. In addition, as we change $c$ from $0$ to $0.1$ there is a rapid decrease in the optimal proportion from $q=1$ to $q=0.45$ and $q=0.4$ in the S-MAP and the IS-MAP procedures, respectively. This result implies that even a small positive weight on the ANO leads to a much smaller proportion value than $q=1$ for which the lowest weighted risk is obtained among the considered proportions. 

\begin{figure}[ht!]
\centering
\includegraphics[scale = 0.42]{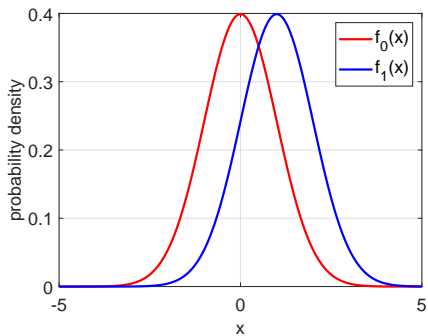}
\caption{Gaussian distributions: Probability densities $f_0 = \mathcal{N}(0,1)$ and $f_1 = \mathcal{N}(1,1)$.}
\label{pdf_Gaussian}
\end{figure}

\begin{figure}[ht!]
\centering
\includegraphics[scale = 0.42]{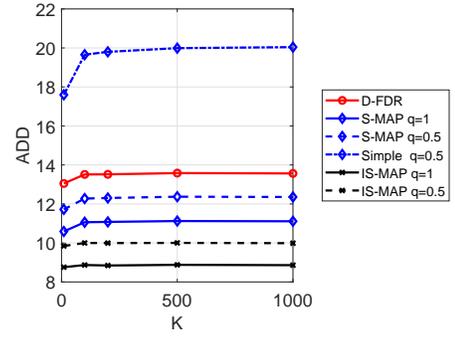}
\caption{Gaussian distributions: The ADD of the procedures D-FDR, S-MAP with $q=0.5,1$, simple, and IS-MAP with $q=0.5,1$, versus $K=10,100,200,500,1000$. It can be seen that all the considered procedures have an approximately constant ADD as $K$ increases. The parallel version of the IS-MAP procedure has the lowest ADD. Moreover, it can be seen that the IS-MAP procedure with $q=0.5$ outperforms the parallel version of the S-MAP procedure and the D-FDR procedure. These results demonstrate the improved ADD performance of the IS-MAP procedure compared to the S-MAP and the D-FDR procedures. The simple procedure with $q=0.5$ has the highest ADD implying that the MAP approach is more useful than random choice when choosing the subset of sensors to monitor.}
\label{ADD_Gaussian}
\end{figure}

\begin{figure}[ht!]
\centering
\includegraphics[scale = 0.42]{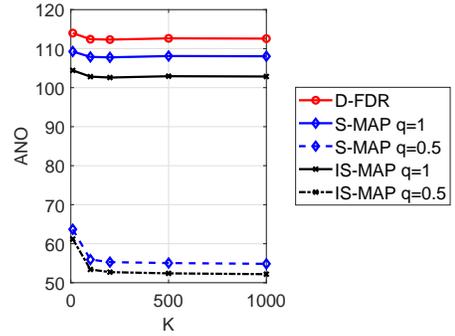}
\caption{Gaussian distributions: The ANO of the procedures D-FDR, S-MAP with $q=0.5,1$, and IS-MAP with $q=0.5,1$, versus $K$. It can be seen that IS-MAP with $q=0.5$ has the lowest ANO, while the D-FDR has the highest ANO. For all the considered procedures, the ANO is approximately a constant w.r.t. $K$.}
\label{ANO_Gaussian}
\end{figure}

\begin{figure}[ht!]
\centering
\includegraphics[scale = 0.42]{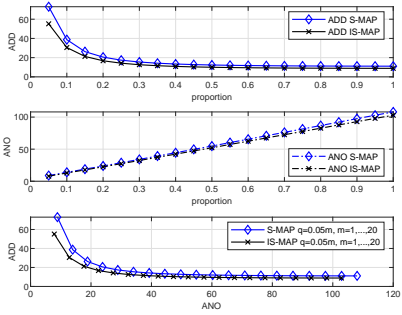}
\caption{Gaussian distributions: (Top and middle) The ADDs and ANOs of the S-MAP and the IS-MAP procedures for $K=1000$ versus the proportion values $\{q=0.05m\}_{m=1}^{20}$. For any of the considered proportions, the IS-MAP procedure has lower ADD and ANO compared to the S-MAP procedure. For both procedures the ADD decreases as the proportion increases, while the ANO increases approximately linearly as the proportion increases. It can be seen that there is no significant increase in ADD when the proportion decreases from $q=1$ to $q=0.3$, whereas the decrease in ANO is more substantial. Thus, with a small proportion we may attain a sufficiently small ADD and significantly decrease the communication burden. (Bottom) A curve connecting the ADD-ANO points of the S-MAP and the IS-MAP procedures from the upper and middle plots of this figure. In this example, there is a clear tradeoff between the ADD and the ANO, i.e. as the the ADD is lower the ANO is higher.}
\label{ADD_ANO_prop_tradeoff_Gaussian}
\end{figure}

\begin{figure}[ht!]
\centering
\includegraphics[scale = 0.42]{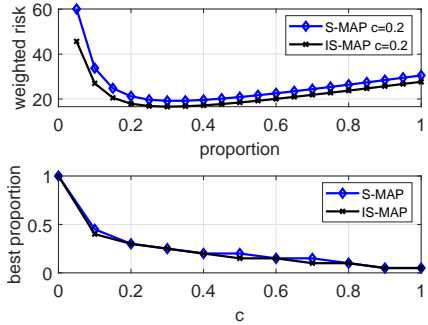}
\caption{Gaussian distributions:  (Top) The weighted risks of the S-MAP and the IS-MAP procedures with proportions $\{q=0.05m\}_{m=1}^{20}$ versus the proportion values for $c=0.2$. The weighted risk of the IS-MAP procedure is lower than the weighted risk of the S-MAP procedure. For both the S-MAP and the IS-MAP procedures, the best tradeoff among the considered proportions is achieved with the proportion $q=0.3$. Thus, in this example, when both the ADD and the ANO are taken into account it is not necessary to monitor all the active data streams in parallel and a lower weighted risk can be attained with a much lower proportion. (Bottom) The best proportion among the proportions $\{q=0.05m\}_{m=1}^{20}$ of the S-MAP and the IS-MAP procedures in terms of the weighted risk in \eqref{weighted_risk} is presented versus the weighting coefficient $c$. For both procedures, as $c$ increases the best proportion value decreases, or stays the same in a few cases. Thus, as we put a higher weight on the ANO compared to the ADD we should usually choose a lower proportion of data streams to observe for achieving a lower weighted risk. In addition, there is a rapid decrease in the optimal proportion from $q=1$ to $q=0.45$ (S-MAP) and $q=0.4$ (IS-MAP), as we change $c$ from $0$ to $0.1$ and thus, it can be seen that even a small positive weight on the ANO leads to a small proportion value, $q<0.5$, for which the lowest weighted risk is obtained among the considered proportions.}
\label{weighted_optimal_risk_Gaussian}
\end{figure}

\subsection{General model with uncertainty and known $p$-values}\label{subsec:General model with uncertainty and known p-values}
Due to bandwidth limitations, in many distributed detection applications the sensors communicate to the FC condensed information about their observations in the form of a local decision and/or sufficient statistic. In this case, significantly less data needs to be communicated. Moreover, the local distributions at each sensor may be different and local decision statistics from each sensor may be easier to fuse than fusing the raw data from all the sensors. A common local decision statistic is the $p$-value \cite{EFRON_BOOK,EFRON_LARGE,HALME}, which is the probability of obtaining test results at least as extreme as the results observed during the test assuming that the null hypothesis is correct. The $p$-value is general and is not necessarily obtained from the Gaussian distribution. It is a tool for deciding whether to reject the null hypothesis. When the $p$-value approaches zero, it is more likely that the alternative hypothesis is true \cite[p. 63]{MOULIN_BOOK}, \cite{BAILEY}.\\
\indent
In this example, we assume that the $p$-values are accurately calculated by each sensor based on its local observations. The $p$-values from each sensor are communicated to the FC for the multiple change-points detection. Under the null hypothesis the $p$-value is uniformly distributed on $[0,1]$ and thus, we set $f_0 = U(0,1)$. Usually, under the alternative hypothesis the $p$-value follows a distribution that has high density for small $p$-values and the density decreases as the $p$-values increase towards $1$ \cite{POUNDS,HEARD}. A commonly assumed distribution for the $p$-value under the alternative hypothesis is the beta distribution \cite{POUNDS,HALME,HEARD}. Therefore, we set $f_1^{(b_k)}=\beta(1,b_k)$, i.e. $f_1^{(b_k)}(X)\propto(1-X)^{b_k-1},~X\in[0,1]$, and $f_1^{(b_k)}(X)=0,~X\notin[0,1]$, where $b_k$ is a parameter of the $k$th data stream probability density under the alternative hypothesis, $\forall k\in[K]$. For each sensor, We consider uncertainty in the value of the parameter $b_k$, where it is only known that $b_k\in[b_{\text{min}},b_{\text{max}}],~\forall k\in[K]$, and $b_{\text{min}},b_{\text{max}}$ are known. The true and unknown value of $b_k$ for each sensor is set by randomly choosing a number in the interval $[b_{\text{min}},b_{\text{max}}]$.\\
\indent
Due to the uncertainty in $f_1^{(b_k)},~k\in[K]$, we implement all the procedures in this example with a generalized LR (GLR), $L_G(X)=\frac{\underset{b\in[b_{\text{min}},b_{\text{max}}]}{\max}f_1^{(b)}(X)}{f_0(X)}$, instead of the actual LR, where we set $b_{\text{min}}=10$ and $b_{\text{max}}=20$. For each data stream, given the observation we compute the corresponding GLR and use its value instead of the unknown LR. The probability densities, $f_0$ and $f_1$ with $b=b_{\text{min}}=10$ and $b=b_{\text{max}}=20$, are depicted in Fig. \ref{pdf_Pvals}. It should be noted that since the true $f_1^{(b_k)},~k\in[K]$, is smaller than or equal to $\underset{b\in[b_{\text{min}},b_{\text{max}}]}{\max}f_1^{(b)}$, the true LR is smaller than the implemented GLR and thus, the resulting FDR may be higher than the predefined upper bound.\\
\indent
We perform similar simulations as in Subsection \ref{subsec:Gaussian distribution scenario}. For $K=10,100,200,500,1000$, we examine the FDR values of the proposed S-MAP and IS-MAP procedures with different proportions $\{q=0.05m\}_{m=1}^{20}$. The resulting minimum and maximum estimated FDR values of the S-MAP procedure are $0.034$ and $0.059$, respectively. The resulting minimum and maximum estimated FDR values of the IS-MAP procedure are $0.064$ and $0.102$, respectively. Consequently, due to the model uncertainty and the maximization of $f_1^{(b)}$ w.r.t. $b\in[b_{\text{min}},b_{\text{max}}]$, some of the resulting FDR values of the IS-MAP procedure are slightly higher than $\alpha=0.1$. This result demonstrates that since the S-MAP procedure is more conservative than the IS-MAP procedure in terms of FDR control then, the S-MAP procedure can be viewed as more robust than the IS-MAP procedure under the assumed model uncertainty.
\begin{remark}
In order to attempt to still maintain the FDR control of the IS-MAP procedure under the desired upper bound, we also implement it with $\rho_{\text{sim}}=0.005$, which is lower than the true value, $\rho=0.01$, under which the random change points are generated. As previously explained, in this case the FDR of the IS-MAP procedure will be lower at the expense of higher ADD. The resulting minimum and maximum estimated FDR values of the IS-MAP procedure are $0.035$ and $0.056$, respectively. Thus, all the IS-MAP estimated FDR values are below the predefined upper bound and FDR control is maintained. In addition, it can be seen that alternating the value of $\rho_{\text{sim}}$ compared to the true $\rho$ is a tool for controlling the tradeoff between FDR and ADD in case of model uncertainty.
\end{remark}
\indent
In Fig. \ref{ADD_Pvals}, we evaluate the ADD of the procedures: D-FDR, S-MAP with $q=0.5,1$, simple procedure with $q=0.5$, and IS-MAP with $q=0.5,1$ versus $K=10,100,200,500,1000$. It can be seen that under the model uncertainty, all the considered procedures still have an approximately constant ADD as $K$ increases, which is in accordance with the analytical results in Section \ref{sec:ADD analysis of the S-MAP and the IS-MAP procedures}. The parallel version of the IS-MAP procedure has the lowest ADD. In addition, the IS-MAP procedure with $q=0.5$ outperforms the parallel version of the S-MAP procedure and the D-FDR procedure, demonstrating the advantage of using the IS-MAP procedure rather than the S-MAP or the D-FDR procedures in terms of ADD. The simple procedure with $q=0.5$ have the highest ADD among the considered procedures. Thus, even under the model uncertainty, there is an advantage in terms of ADD in using the proposed MAP approach for choosing the monitored sensors rather than randomly choosing the subset of sensors to monitor. In Fig. \ref{ANO_Pvals}, we evaluate the ANO versus $K$ of the procedures: D-FDR, S-MAP with $q=0.5,1$, and IS-MAP with $q=0.5,1$. It can be seen that IS-MAP with $q=0.5$ has the lowest ANO, whereas the D-FDR has the highest one. In all the considered procedures, the ANO is approximately a constant w.r.t. $K$.\\
\indent
In the upper and middle plots of Fig. \ref{ADD_ANO_prop_tradeoff_Pvals}, we plot the ADDs and ANOs, respectively, of the S-MAP and the IS-MAP procedures for $K=1000$ versus the proportion values $\{q=0.05m\}_{m=1}^{20}$. It can be seen that for any of the considered proportions, the IS-MAP procedure has lower ADD and ANO than the S-MAP procedure. In addition, for both procedures the ADD decreases as the proportion increases, while the ANO increases as the proportion increases. Similar to the previous example, it can be seen that there is no significant increase in ADD when the proportion decreases from $q=1$ to $q=0.3$. The ANO increases significantly as $q$ increases towards $1$. In the lower plot of Fig. \ref{ADD_ANO_prop_tradeoff_Pvals}, we plot a curve connecting the ADD-ANO points of the S-MAP and the IS-MAP procedures from the upper and middle plots of Fig. \ref{ADD_ANO_prop_tradeoff_Pvals}. It can be seen that under the model uncertainty we still have a clear tradeoff between the ADD and ANO and the ADD decreases as the ANO increases.\\
\indent
In the upper plot of Fig. \ref{weighted_optimal_risk_Pvals}, we compare the weighted risks from \eqref{weighted_risk} of the S-MAP and the IS-MAP procedures with proportions $\{q=0.05m\}_{m=1}^{20}$ versus the proportion size for $c=0.2$. It can be seen that the weighted risk of the IS-MAP procedure is lower than the weighted risk of the S-MAP procedure. For both the S-MAP and the IS-MAP procedures, the best tradeoff among the considered proportions is achieved with the proportion $q=0.2$. Thus, under the model uncertainty, it is still not desirable to monitor all the active data streams in parallel, when both ADD and ANO are taken into account. In the lower plot of Fig. \ref{weighted_optimal_risk_Pvals}, for both the S-MAP and the IS-MAP procedures, we present the best proportion among the proportions $\{q=0.05m\}_{m=1}^{20}$ in terms of the weighted risk in \eqref{weighted_risk} versus the weighting coefficient $c$. Similarly to the previous example, for both procedures, as we increase $c$ the best proportion value decreases or does not increase. We also noticed a rapid decrease in the optimal proportion from $q=1$ to $q=0.25$, as we change $c$ from $0$ to $0.1$.

\begin{figure}[ht!]
\centering
\includegraphics[scale = 0.42]{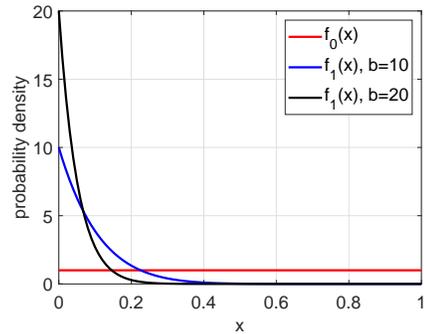}
\caption{General model with uncertainty and known $p$-values: Probability density $f_0 = U(0,1)$ that corresponds to the distribution of a $p$-value under the null hypothesis, $f_1=\beta(1,b),~b=b_{\text{min}}=10,b=b_{\text{max}}=20$. The beta distribution is a common assumption for a $p$-value under the alternative hypothesis. We assume that $b\in[b_{\text{min}},b_{\text{max}}]$ is unknown and depict the corresponding probability densities with the lowest possible value of $b$, $b_{\text{min}}=10$, and the highest possible value of $b$, $b_{\text{max}}=20$.}
\label{pdf_Pvals}
\end{figure}

\begin{figure}[ht!]
\centering
\includegraphics[scale = 0.42]{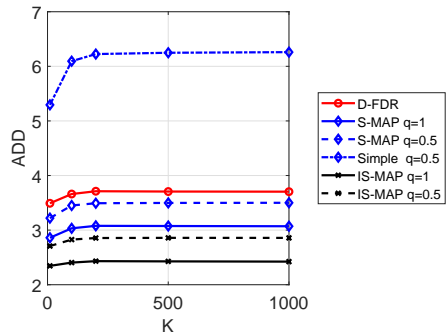}
\caption{General model with uncertainty and known $p$-values: The ADD of the procedures D-FDR, S-MAP with $q=0.5,1$, simple, and IS-MAP with $q=0.5,1$, versus $K=10,100,200,500,1000$. It can be seen that even under the model uncertainty all the considered procedures have an approximately constant ADD as $K$ increases. The parallel version of the IS-MAP procedure has the lowest ADD. Similarly to the previous example, the IS-MAP procedure with $q=0.5$ outperforms the parallel version of the S-MAP procedure and the D-FDR procedure. The simple procedure with $q=0.5$ has the highest ADD. Thus, even under model uncertainty, the MAP approach outperforms a random choice approach for choosing the subset of sensors to monitor.}
\label{ADD_Pvals}
\end{figure}

\begin{figure}[ht!]
\centering
\includegraphics[scale = 0.42]{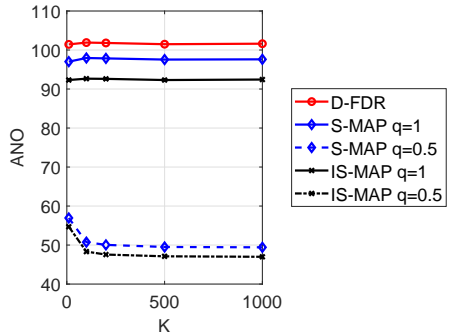}
\caption{General model with uncertainty and known $p$-values: The ANO of the procedures D-FDR, S-MAP with $q=0.5,1$, and IS-MAP with $q=0.5,1$, versus $K$. Under the model uncertainty, all the considered procedures still have an approximately constant ANO w.r.t. $K$, similarly to the previous example in which there is no model uncertainty. The IS-MAP with $q=0.5$ has the lowest ANO, while the D-FDR has the highest ANO.}
\label{ANO_Pvals}
\end{figure}

\begin{figure}[ht!]
\centering
\includegraphics[scale = 0.42]{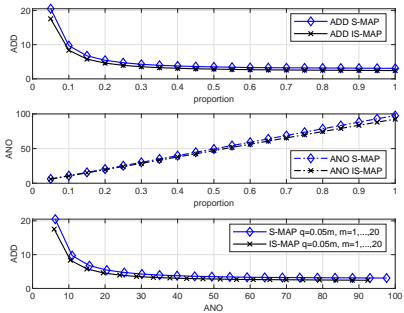}
\caption{General model with uncertainty and known $p$-values: (Top and middle) The ADDs and ANOs of the S-MAP and the IS-MAP procedures for $K=1000$ versus the proportion values $\{q=0.05m\}_{m=1}^{20}$. The IS-MAP procedure has lower ADD and ANO compared to the S-MAP procedure. For both procedures the ADD decreases as the proportion increases, while the ANO increases approximately linearly as the proportion increases. There is no significant increase in ADD when the proportion decreases from $q=1$ to $q=0.3$, while the decrease in ANO is more significant. (Bottom)  A curve connecting the ADD-ANO points of the S-MAP and the IS-MAP procedures from the upper and middle plots of this figure. Under the model uncertainty there is still a clear tradeoff between the ADD and ANO, i.e. as the ADD decreases the ANO increases.}
\label{ADD_ANO_prop_tradeoff_Pvals}
\end{figure}

\begin{figure}[ht!]
\centering
\includegraphics[scale = 0.42]{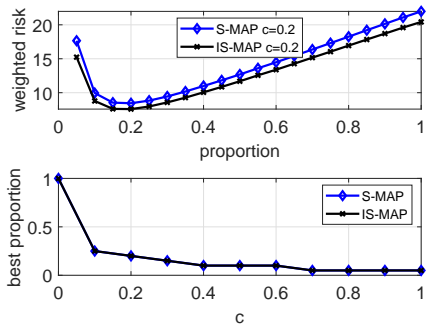}
\caption{General model with uncertainty and known $p$-values: (Top) The weighted risks of the S-MAP and the IS-MAP procedures with proportions $\{q=0.05m\}_{m=1}^{20}$ versus the proportion size for $c=0.2$. The weighted risk of the IS-MAP procedure is lower than the weighted risk of the S-MAP procedure. For both the S-MAP and the IS-MAP procedures, the best tradeoff among the considered proportions is achieved with the proportion $q=0.2$. Thus, in a similar manner to the previous example, when both the ADD and the ANO are taken into account it is not desirable to monitor all the active data streams in parallel and a lower weighted risk can be attained with a much lower proportion. (Bottom) The best proportion among the proportions $\{q=0.05m\}_{m=1}^{20}$ of the S-MAP and the IS-MAP procedures in terms of the weighted risk in \eqref{weighted_risk} is presented versus the weighting coefficient $c$. In this case, the S-MAP and the IS-MAP procedures have the same best proportions. Similar to previous example, as $c$ increases the best proportion value, among the considered proportions, decreases or does not increase. A rapid decrease is noticed in the optimal proportion from $q=1$ to $q=0.25$, as we change $c$ from $0$ to $0.1$.}
\label{weighted_optimal_risk_Pvals}
\end{figure}

\section{Conclusion}\label{sec:Conclusion}
In this paper, we developed methods for Bayesian multiple change-point detection in sensor network with limitations on the proportion of sensors that can be monitored in parallel. We proposed the S-MAP detection procedure in which observations are received only from a subset of sensors with highest posterior probabilities of change points having occurred, within the allowed proportion. In addition, we proposed an improved procedure named the IS-MAP procedure that requires lower thresholds than the S-MAP procedure and attains lower ADD and ANO. It has been shown that both the proposed procedures control the FDR at a predefined level and achieve an ADD that asymptotically remains a constant as the number of sensors in the network increases. The S-MAP procedure is more conservative than the IS-MAP procedure in terms of FDR control, and thus, in the FDR control sense, the S-MAP procedure is more robust to model uncertainty than the IS-MAP procedure. In the simulations, we have first considered i.i.d. Gaussian observations with a change in the mean and then we have considered a general model with some model uncertainty in which $p$-values from each sensor are used as observations to perform the change-points detection task. Our simulations in both cases show that the proposed S-MAP and IS-MAP procedures achieve a practically constant ADD as the number of sensors increases. The S-MAP procedure outperforms a corresponding simple procedure in terms of ADD demonstrating the benefit of the MAP approach compared to randomly choosing the subset of sensors to monitor. We have also used the S-MAP and IS-MAP procedures to study the tradeoff between ADD and ANO in multiple change-point detection. Under a joint weighted risk on the ADD and ANO with a positive weight on both figures of merit, we found that in all the considered cases observing all the data streams, i.e. setting $q=1$, does not provide the best tradeoff between the ADD and ANO. In fact, the best tradeoff can be obtained with proportion $q\ll1$, which implies that setting a small proportion, e.g. $q=0.3$, can significantly reduce the communication burden, i.e. the ANO, while maintaining a low ADD. A Topic for future research is the derivation of novel procedures with FDR control capabilities for non-parametric \cite{LAU_TAY_VEERAVALLI} multiple change-point detection under communication limitations.

\appendices

\section{Proof of Theorem \ref{T_FDR_IS_MAP}}\label{App_T_FDR_IS_MAP}
In this appendix, the FDR control of the IS-MAP procedure is proved. The number of change points declared, $R$, is known given the filtration of all the data, ${\mathcal{F}}_\infty$. Thus, using the law of total expectation, we can rewrite the FDR from \eqref{FDR_define} as
\be\label{FDR_condition}
{\text{FDR}}={\rm{E}}\bigg[\frac{{\rm{E}}[V|{\mathcal{F}}_\infty]}{\max\{ R,1\}}\bigg].
\ee
Recall that $V$ is the number of false discoveries, i.e. the size of the subset of $[K]$ s.t. $T^{(k)}<t^{(k)}$. Thus, $V$ can be written as
\be\label{V_with_indicators}
V=\sum_{k=1}^{K}\onevec_{\{t^{(k)}>T^{(k)}\}},
\ee
where $\onevec_{A}$ is the indicator function of the event $A$. By substituting \eqref{V_with_indicators} in \eqref{FDR_condition} and using the linearity of the expectation operator, we obtain
\be\label{FDR_rewriting}
\begin{split}
{\text{FDR}}&={\rm{E}}\bigg[\frac{\sum_{k=1}^{K}{\rm{E}}[\onevec_{t^{(k)}>T^{(k)}}|{\mathcal{F}}_\infty]}{\max\{R,1\}}\bigg]\\
&={\rm{E}}\bigg[\frac{\sum_{k=1}^{K}{\rm{E}}[\onevec_{t^{(k)}>T^{(k)}}|{\mathcal{F}}_{T^{(k)}}]}{\max\{ R,1\}}\bigg],
\end{split}
\ee	
where the second equality is obtained since the stopping times, $\{T^{(k)}\}_{k\in[K]}$, are known given ${\mathcal{F}}_\infty$ and for $T^{(k)}<\infty$ we stop observing the $k$th data stream after $T^{(k)}$, i.e. after change point declaration for the $k$th data stream. Rewriting the expected indicator functions in \eqref{FDR_rewriting} as conditional probabilities, we obtain
\be\label{FDR_rewriting_next}
\begin{split}
{\text{FDR}}&={\rm{E}}\bigg[\frac{\sum_{k=1}^{K}(1-P(t^{(k)}\leq T^{(k)}|{\mathcal{F}}_{T^{(k)}}))}{\max\{ R,1\}}\bigg]\\
&={\rm{E}}\bigg[\frac{\sum_{k=1}^{K}(1-\pi_{T^{(k)}}^{(k)})}{\max\{ R,1\}}\bigg],
\end{split}
\ee	
where the second equality is obtained by substituting \eqref{posterior_prob} into the first equality. In case $T^{(k)}=\infty$, then $\pi_{T^{(k)}}^{(k)}=1$ and thus, 
\be\label{expected_event_k_inf}
1-\pi_{T^{(k)}}^{(k)}= 0,~\forall k\in[K],~{\text{s.t.}}~T^{(k)}=\infty.
\ee
On the other hand, in case $T^{(k)}<\infty$ then, at time slot $T^{(k)}$ the event $\{\pi_{T^{(k)}}^{(k)}\geq Q\}$ occurs. Consequently,
\be\label{expected_event_k}
1-\pi_{T^{(k)}}^{(k)}\leq 1-Q = \alpha,~\forall k\in[K],~{\text{s.t.}}~T^{(k)}<\infty,
\ee
where the last equality is obtained by substituting $Q=1-\alpha$ from \eqref{thresholds_eq_alpha} into the term $1-Q$. The term $R$ is the cardinality of the subset of $[K]$ s.t. $T^{(k)}<\infty$. Thus, by substituting \eqref{expected_event_k_inf}-\eqref{expected_event_k} in \eqref{FDR_rewriting_next}, one obtains
\be\label{FDR_rewriting_final}
{\text{FDR}}\leq{\rm{E}}\bigg[\frac{\max\{ R,1\}\alpha}{\max\{ R,1\}}\bigg]=\alpha.
\ee	

\section{Proof of Theorem \ref{bounds_ADD}}\label{App_bounds_ADD}
In this appendix, we derive asymptotic lower and upper bounds on the ADD of the S-MAP and the IS-MAP procedure. For any data stream, the lowest possible threshold of the S-MAP procedure from \eqref{high_thresholds} is $Q_K=1-\alpha$, i.e. change point cannot be declared before the posterior probability is higher than or equal to $1-\alpha$. Thus, from \eqref{single_lower_bound}
\be\label{single_lower_bound_k}
{\text{ADD}}_{\text{S-MAP},k}\geq\frac{|\log\alpha|}{D(f_1||f_0)+|\log(1-\rho)|}(1+o_\alpha(1)),
\ee
$\forall k\in[K]$. It can be seen that the asymptotic lower bound in \eqref{single_lower_bound_k} is independent of $k$. Thus, by substituting \eqref{single_lower_bound_k} in \eqref{ADD}, we obtain \eqref{asympt_lower_bound_S_MAP}.\\
\indent
According to the S-MAP procedure we can find a threshold for the $k$th data stream, $Q_{r_k}=1-\frac{r_k\alpha}{K},~r_k\in[K]$, which is different from the thresholds of the other data streams. For this threshold, the change of the $k$th data stream is declared at the first time slot in which this threshold is exceeded or even before this threshold is exceeded. Thus, from \eqref{single_upper_bound},
\be\label{single_upper_bound_k}
{\text{ADD}}_{\text{S-MAP},k}\leq\frac{\log\frac{K}{r_k\alpha}}{|\log(1-\rho)|}(1+o_\alpha(1)),\forall k\in[K].
\ee
By substituting \eqref{single_upper_bound_k} in \eqref{ADD}, we obtain
\be\label{ADD_upper_bound}
{\text{ADD}}_{\text{S-MAP}}\leq\left(\frac{1}{K}\sum_{k=1}^{K}\frac{\log\frac{K}{r_k\alpha}}{|\log(1-\rho)|}\right)(1+o_\alpha(1)).
\ee
Since the thresholds are different, we obtain 
\be\label{thresholds_permutation}
\sum_{k=1}^{K}\log r_k=\sum_{k=1}^{K}\log k=\log K!.
\ee
By substituting \eqref{thresholds_permutation} into \eqref{ADD_upper_bound} and reordering, we obtain \eqref{asympt_upper_bound_S_MAP}.\\
\indent
In the IS-MAP procedure, for any data stream the threshold is $Q=1-\alpha$ from \eqref{thresholds_eq_alpha}. Thus, using \eqref{single_lower_bound} and \eqref{single_upper_bound}, we obtain
\be\label{single_lower_bound_k_D}
{\text{ADD}}_{\text{IS-MAP},k}\geq\frac{|\log\alpha|}{D(f_1||f_0)+|\log(1-\rho)|}(1+o_\alpha(1))
\ee
and
\be\label{single_upper_bound_k_D}
{\text{ADD}}_{\text{IS-MAP},k}\leq\frac{|\log\alpha|}{|\log(1-\rho)|}(1+o_\alpha(1)),
\ee
respectively, $\forall k\in[K]$. The asymptotic lower and upper bounds in \eqref{single_lower_bound_k_D} and \eqref{single_upper_bound_k_D}, respectively, are independent of $k$ and thus, by substituting \eqref{single_lower_bound_k_D} and \eqref{single_upper_bound_k_D} in \eqref{ADD}, we obtain \eqref{asympt_lower_bound_IS_MAP} and \eqref{asympt_upper_bound_IS_MAP}, respectively.

\section{Proof of Proposition \ref{single_ADD}}\label{App_single_ADD}
In this appendix, we derive the asymptotic ADD upper bound from \eqref{single_upper_bound_assumption} under the assumption that \eqref{bounded_interval_define}-\eqref{convergent_expectation_mean} are satisfied. Using the definition of $\gamma$ from \eqref{single_change_point_ignore}, it can be seen that the prior distribution of $\gamma\in\mathbb{N}$ is
\be\label{prior_prob_prop}
\begin{split}
P(\gamma=m)&=P(V_{m-1} <t\leq V_m)\\
&=P(t\leq V_m)-P(t\leq V_{m-1}).
\end{split}
\ee
Under the geometric prior assumption on $t$ we obtain
\be\label{geom_property}
P(t\leq m)=1-(1-\rho)^m,m\in\mathbb{N}.
\ee
By substituting \eqref{geom_property} in \eqref{prior_prob_prop}, one obtains
\be\label{prior_prob_prop_next}
P(\gamma=m)=(1-\rho)^{V_{m-1}}-(1-\rho)^{V_{m}}.
\ee
Using \eqref{prior_prob_prop_next}, we obtain
\be\label{general_exponential_tail}
\begin{split}
\underset{m\to\infty}{\lim}\frac{-\log P(\gamma\geq m+1)}{m}&=\underset{m\to\infty}{\lim}\frac{-\log ((1-\rho)^{V_{m}})}{m}\\
&=\bigg(\underset{m\to\infty}{\lim}\frac{V_m}{m}\bigg)|\log(1-\rho)|\\
&=\zeta|\log(1-\rho)|,
\end{split}
\ee
where the third equality is obtained by substituting \eqref{arithmetic_mean} and \eqref{convergent_interval_mean} into the second equality. Using the definition of $\gamma$ from \eqref{single_change_point_ignore}, we obtain that on $\{\gamma=n\}$
\be\label{KLD_converge}
\underset{N\to\infty}{\lim}\frac{1}{N}\sum_{i=n}^{n+N-1}\log L(X_{V_i})=D(f_1||f_0).
\ee
From the definitions of the stopping rule, $\Gamma$, and the change point, $\gamma$, in \eqref{single_stop_rule_ignore} and \eqref{single_change_point_ignore}, respectively, and from \eqref{general_exponential_tail} and \eqref{KLD_converge}, it can be seen that the detection of $\gamma$ using $\Gamma$ based on the sequence $\{X_{V_n}\}_{n=1}^{\infty}$ is a Bayesian change-point detection procedure that satisfies the conditions of Theorem 3 in \cite{TARTAKOVSKY_GENERAL}. Thus, using this Theorem, we obtain the following asymptotic upper bound on the ADD of $\Gamma$:
\be\label{ADD_ignore_procedure}
{\rm{E}}[\max\{0,\Gamma-\gamma\}]\leq\frac{|\log\eta|}{D(f_1||f_0)+\zeta|\log(1-\rho)|}(1+o_\eta(1)).
\ee
Next, we consider the stopping rule
\be\label{single_stop_rule_star}
T^*=\inf\{V_n,n\in\mathbb{N}:\pi_{V_n}\geq 1-\eta\}=V_\Gamma.
\ee 
In a similar manner to $T$, the stopping rule $T^*$ uses the posterior update from \eqref{recursive_update}, but can only take values from the subsequence $\{V_n\}_{n=1}^\infty$ rather than $\mathbb{N}$. Therefore, $T\leq T^*$ and consequently
\be\label{add_diff_develop}
T-t\leq T^*-t=V_\Gamma-V_\gamma+V_\gamma-t,
\ee
where the equality follows from \eqref{single_stop_rule_star}. From \eqref{interval_define} and \eqref{single_change_point_ignore} we obtain that
\be\label{gamma_little_property}
V_\gamma-t\leq \zeta_\gamma-1.
\ee
In addition, using \eqref{arithmetic_mean} we can write
\be\label{V_Gamma_eq}
V_\Gamma=\Gamma\zeta^{(\Gamma)}~\text{and}~V_\gamma=\gamma\zeta^{(\gamma)}.
\ee
By substituting \eqref{gamma_little_property}-\eqref{V_Gamma_eq} into the right hand side of \eqref{add_diff_develop}, one obtains
\be\label{add_diff_sub_gamma_ver8}
\begin{split}
T-t&\leq \zeta^{(\Gamma)}(\Gamma-\gamma)+\gamma(\zeta^{(\Gamma)}-\zeta^{(\gamma)})+\zeta_\gamma-1\\
&\leq \zeta^{(\Gamma)}(\Gamma-\gamma)+|\gamma(\zeta^{(\Gamma)}-\zeta^{(\gamma)})+\zeta_\gamma-1|.
\end{split}
\ee
Using \eqref{interval_define} and \eqref{bounded_interval_define}, we obtain
\be\label{arithmetic_mean_app}
1\leq\zeta^{(N)}\leq{\mathcal{B}},~\forall N\in{\mathbb{N}}.
\ee
Substituting \eqref{bounded_interval_define} and \eqref{arithmetic_mean_app} in \eqref{add_diff_sub_gamma_ver8}, one obtains
\be\label{add_diff_sub_gamma}
T-t\leq \zeta^{(\Gamma)}(\Gamma-\gamma)+\gamma({\mathcal{B}}-1)+{\mathcal{B}}-1.
\ee
From \eqref{prior_prob_prop_next}, it can be verified that
\be\label{expected_prior_prob_prop}
{\rm{E}}[\gamma]\leq\frac{1}{\rho}.
\ee
By using \eqref{convergent_expectation_mean}, \eqref{ADD_ignore_procedure}, \eqref{add_diff_sub_gamma} and \eqref{expected_prior_prob_prop}, we obtain that the ADD of $T$ satisfies
\be\label{single_upper_bound_zeta}
{\text{ADD}}\leq\frac{\zeta|\log\eta|}{D(f_1||f_0)+\zeta|\log(1-\rho)|}(1+o_\eta(1))
\ee
and consequently \eqref{single_upper_bound_assumption} is obtained.


\bibliographystyle{IEEEtran}
\bibliography{refs}

\end{document}